\documentclass{fundam}

\usepackage[utf8]{inputenc}
\usepackage[T1]{fontenc}

\usepackage{url} 
\usepackage[colorlinks = true, linkcolor = black, urlcolor = black, citecolor = black, anchorcolor = black]{hyperref}

\usepackage{mathtools}
\usepackage{amssymb}
\usepackage{enumitem}

\newcommand{\NN}{\mathbb{N}}
\newcommand{\AAA}{\mathbb{A}}
\newcommand{\BBB}{\mathbb{B}}
\newcommand{\CCC}{\mathbb{C}}
\newcommand{\CA}{\mathcal{A}}
\newcommand{\CB}{\mathcal{B}}
\newcommand{\CC}{\mathcal{C}}
\newcommand{\CD}{\mathcal{D}}
\newcommand{\CI}{\mathbb{I}}
\newcommand{\CN}{\mathcal{N}}
\newcommand{\CT}{\mathcal{T}}
\newcommand{\CU}{\mathcal{U}}
\newcommand{\CV}{\mathcal{V}}
\newcommand{\CW}{\mathcal{W}}

\newcommand{\interleave}[2]{{{#1} \textsuperscript{$\wedge$} {#2}}}

\makeatletter
    \@ifdefinable{\pOne}{\def\pOne/{\text{I}}}
    \@ifdefinable{\pTwo}{\def\pTwo/{\text{II}}}
\makeatother

\begin{document}


\setcounter{page}{63}
\publyear{22}
\papernumber{2119}
\volume{186}
\issue{1-4}

  \finalVersionForARXIV


\title{Solving Infinite Games in the Baire Space}

\author{Benedikt Br\"utsch\thanks{Address for correspondence:  Chair of Computer Science 7,
                       RWTH Aachen University, Germany.}, Wolfgang Thomas\\
    Chair of Computer Science 7\\
    RWTH Aachen University, Germany\\
    bruetsch@automata.rwth-aachen.de, thomas@cs.rwth-aachen.de
}

\maketitle

\runninghead{B. Br\"utsch and W. Thomas}{Solving Infinite Games in the Baire Space}


\begin{abstract}
    Infinite games (in the form of Gale-Stewart games) are studied where a play is a sequence of natural numbers chosen by two players in alternation, the winning condition being a subset of the Baire space $\omega^\omega$.
    We consider such games defined by a natural kind of parity automata over the alphabet $\NN$, called $\NN$-MSO-automata, where transitions are specified by monadic second-order formulas over the successor structure of the natural numbers.
    We show that the classical B\"uchi-Landweber Theorem (for finite-state games in the Cantor space $2^\omega$) holds again for the present games:
    A game defined by a deterministic parity $\NN$-MSO-automaton is determined, the winner can be computed, and an $\NN$-MSO-transducer realizing a winning strategy for the winner can be constructed.
\end{abstract}

\begin{keywords}
    infinite games, Gale-Stewart games, determinacy, automata over infinite alphabets, transducers, monadic second-order logic, Church's synthesis problem, pushdown systems
\end{keywords}

\section{Introduction}
\label{sec:introduction}

The subject of infinite games (in the sense of turn-based two-player games with complete information) has developed in two fields, descriptive set theory and automata theory.
In descriptive set theory, the focus is the problem of determinacy (i.e., the question whether there is a winning strategy for either of the two players), whereas in automata theory the decision who wins and the construction of winning strategies are the main issues.
It is interesting to note that in descriptive set theory it is not essential whether plays of an infinite game are infinite sequences of bits (or of letters from a finite alphabet), i.e., elements of the Cantor space $2^\omega$, or sequences of natural numbers, i.e., elements of the Baire space $\omega^\omega$.
The results on determinacy, in particular the determinacy of Borel games, hold for both cases, shown with the same proof strategy.
This is different in automata theory where a game (more precisely, its winning condition) is defined by an automaton processing infinite sequences over a finite alphabet and where  the treatment of infinite alphabets (such as $\NN$) causes some difficulties.
Various models of  automata over infinite alphabets have been introduced, mainly in the context of database theory (and then considered over finite sequences), but none of these models can be considered as ``canonical'', and it seems that such automata have so far not been used in the study of infinite games.

The purpose of the present paper is to show that a central result of the automata-theoretic approach to infinite games can be lifted from the Cantor space to the Baire space.
It is the ``B\"uchi-Landweber Theorem''~\cite{BL69}, stating that a game defined by a regular $\omega$-language is determined in the strong sense that the winner can be computed and a finite-state transducer realizing a winning strategy for the winner can be constructed.
Over infinite alphabets we use here a specific (but quite natural) automata-theoretic form of specification of $\omega$-languages, namely by deterministic acceptors that are equipped with a memory consisting of two parts, a finite set of control states and additional memory that can store a natural number.
Transitions are defined by MSO-formulas over the structure $\CN_\mathrm{succ}$ (consisting of the set $\NN$ and the successor function $\mathrm{succ}\colon n \mapsto n+1$).
When these automata are used with the parity condition for the acceptance of infinite words, we call them parity $\NN$-MSO-automata.
For solutions of games we use corresponding transducers, called $\NN$-MSO-transducers.
We show that a game that is defined by a deterministic parity $\NN$-MSO-automaton $\CA$ over the alphabet $\NN$ is determined in the strong sense that the winner can be computed and an $\NN$-MSO-transducer realizing a winning strategy for the winner can be constructed (i.e., computed) from $\CA$.

The first monograph in which the original B\"uchi-Landweber Theorem appeared was the English edition of a truly pioneering book: \emph{Finite Automata -- Behavior and Synthesis} by  Trakhtenbrot and Barzdin~\cite{TB73}.
In the present paper, commemorating the 100th birthday of Boaz Trakhtenbrot, we take up this title and proceed in two stages, first on ``behavior'', introducing the framework of game specification in terms of automata, second on ``synthesis'', where the main result on construction of winning strategies in terms of corresponding transducers is given.

We assume that the reader is acquainted with the fundamentals of automata on infinite words, infinite games, and monadic second-order logic, as presented, e.g., in~\cite{GrThWi, Tho97}.

The main result of this paper was obtained in joint work during the research for the first author's PhD dissertation~\cite{Bru18}.
An early account was given in~\cite{BT16}, proposing a construction of MSO-definable strategies for parity games from MSO-definable strategies for reachability games.
A gap in the justification of this construction motivated the approach pursued in the present paper (and in~\cite{Bru18}) that establishes MSO-definable strategies directly for parity games.

Results on (variants of) $\NN$-MSO-automata were presented in~\cite{speltenwinterthomas, CST15, BLT17}.

We thank the reviewer for his/her most careful reading and numerous valuable suggestions to improve the paper.

\section{Behavior}

\subsection{\texorpdfstring{$\NN$}{N}-MSO-automata}

We consider automata over the alphabet $\NN$.
These automata are equipped with a finite set $Q$ of control states and the set $\NN$ as additional memory;
so a ``configuration'' is a pair $(p,i) \in Q \times \NN$.
Often we use $(q_0,0)$ as initial configuration where $q_0$ is a designated initial control state.
A transition of an $\NN$-MSO-automaton is a triple $((p,i), m, (q,j))$ leading from configuration $(p,i)$ via an input number $m \in \NN$  to configuration $(q,j)$.
We define for each pair $(p,q)$ of control states the set of possible transitions by an MSO-formula  $\varphi_{pq}(x, z, y)$ with first-order variables $x,y,z$.
The transition $((p,i), m, (q,j))$ is allowed if
\begin{equation}\label{transition}
\CN_\mathrm{succ}  \models \varphi_{pq}[i, m , j].  \tag{$\ast$}
\end{equation}

\begin{remark}\label{automatic}
    The ternary relation defined by a formula $\varphi_{pq}$ can as well be defined by a finite automaton processing a triple $(i,m,j)$ in unary notation, i.e., as the word triple $(1^i, 1^m, 1^j)$ scanned in parallel letter by letter (with a dummy symbol used at the end, up to the length of the longest component).
    So a transition formula defines an automatic relation over the singleton alphabet $\{1\}$.
\end{remark}

For the acceptance of finite words, we use a subset $F \subseteq Q$ as acceptance component,
whereas for the acceptance of infinite words, we use a coloring $c\colon Q \rightarrow \{0, \ldots, d\}$ of the control states by natural numbers.
So an $\NN$-MSO-automaton is of the form
$\CA = (Q, q_0, (\varphi_{pq}(x, z, y))_{p,q \in Q}, \mathrm{Acc})$
where $\mathrm{Acc}$ is either a set $F \subseteq Q$ or a coloring  $c\colon Q \rightarrow \{0, \ldots, d\}$ of states by natural numbers.
A ``machine-oriented'' format of these automata, avoiding the reference to MSO-logic, is given at the end of this subsection.

A \emph{run} of an $\NN$-MSO-automaton $\CA$ over a finite word $w = m_0 \ldots m_\ell$ is a finite sequence of configurations $(p_0,0), (p_1, i_1), \ldots (p_{\ell+1}, i_{\ell +1})$
such that $p_0 = q_0$ and
$\CN_\mathrm{succ}  \models \varphi_{p_{k} p_{k+1}}[i_k, m_k, i_{k+1}]$
for all $k \in \{0,\dotsc,\ell\}$.
We say that $\CA$ \emph{accepts} $w$ if such a run exists with $p_{\ell+1} \in F$.
Analogously, a run of an $\NN$-MSO-automaton $\CA$ over an infinite word $\alpha = m_0 m_1 m_2 \ldots$ is an infinite sequence of configurations $(p_0,0), (p_1, i_1), (p_2, i_2), \ldots$
such that $p_0 = q_0$ and
$\CN_\mathrm{succ}  \models \varphi_{p_{k} p_{k+1}}[i_k, m_k, i_{k+1}]$
for all $k \in \NN$.
We say that $\CA$ accepts $\alpha$ (by the ``parity condition'') if there is a corresponding run such that among the colors $c(p_0), c(p_1), \ldots$ the maximal color occurring infinitely often is even;
used with this acceptance condition for infinite words, $\CA$ is called a \emph{parity $\NN$-MSO-automaton}.
We write $L_\CA$ to denote the language of words accepted by $\CA$ (where the type of acceptance component of $\CA$ determines whether this language consists of finite or infinite words) and we call such languages $\NN$-MSO-\emph{recognizable}.

\medskip
In \emph{deterministic $\NN$-MSO-automata} we use formulas $\varphi_{pq}(x,z,y)$ that are functional in $p$, $x$, and $z$, i.e., such that for each $p$, $i$, $m$, there is exactly one pair $(q,j)$ with (\ref{transition}).
This is a decidable condition, since one can formalize it as the MSO-sentence
\[
    \bigwedge_p (\forall x \forall  z \bigvee_q [\exists^{=1}y\, \varphi_{pq}(x, z, y)
    \wedge \bigwedge_{r \not= q} \forall t \neg \ \varphi_{pr}(x,z, t)])
\]
and apply decidability of the MSO-theory of~$\CN_\mathrm{succ}$.

Let us present some simple examples.
First we consider a language of  finite words, namely those of the form $n (n+1) (n+2) \ldots (n+\ell)$ with $\ell > 0$, which can be recognized by a two-state $\NN$-MSO-automaton.
By the first transition from control state $q_0$ to $q$, the memory is set to the first input value~$n$, using the formula $\varphi_{q_0 q}(x,z,y) \coloneqq (x=0) \land (y=z)$.
Subsequently, now in state $q$, it is checked that the current input $m$ is the increment by 1 of the current memory content, and the memory content is increased by 1 to become $m$.
This is done by the formula $\varphi_{q q}(x,z,y) \coloneqq (z=x+1) \land (y=z)$.
We set $F =\{q\}$.
After completion by transitions for the rejection of input words, we obtain a deterministic automaton.

The second example is an $\omega$-language, namely the set of unbounded sequences over $\NN$.
We describe how the corresponding parity $\NN$-MSO-automaton should work:
It keeps in its memory, starting with value $0$, the maximum of the numbers seen so far, and adjusts this value whenever this maximum has to be increased.
In these steps the control state receives color 2, otherwise the color is 1.
Again we obtain a deterministic automaton.

We use nondeterminism in our third example, the $\omega$-language of those sequences in which some number occurs infinitely often.
Here the automaton guesses an occurrence of such a number, uses its memory to propagate this value, and signals by visiting a state of color 2 whenever this value is encountered again; other states have color 1.

In \emph{$\NN$-MSO-transducers}, the acceptance component of deterministic $\NN$-MSO-automata as introduced above is replaced by an output component that associates with each configuration a number $n$ as output.
This output component is specified by a family of MSO-formulas $\psi_{q}(x, z)$, for $q \in Q$, which are functional in $q$ and $x$, i.e., satisfying the condition
$\CN_\mathrm{succ} \models \bigwedge_q \forall x \exists^{=1}z\, \psi_q (x, z)$.
So the format of a transducer is
$\CC = (Q, q_0, (\varphi_{pq}(x, z, y))_{p,q \in Q}, (\psi_{q}(x,z))_{q \in Q})$.
For a configuration $(q,j)$ reached via the current input transition formula $\varphi_{pq}(x,z,y)$, the output is the unique number $n$ with $\CN_\mathrm{succ} \models \psi_q[j, n]$.

An $\NN$-MSO-transducer $\CA$ defines in this way a function $F_\CA\colon \NN^* \rightarrow \NN^*$ where $F_\CA(w)$ is obtained as the sequence of outputs of the configurations of the unique run of $\CA$ on $w$, with the convention that the initial configuration yields no output.
When considering infinite words, we denote by $F_\CA$ the function $F_\CA\colon \NN^\omega \rightarrow \NN^\omega$ where $F_\CA(\alpha)$ is obtained analogously from the unique run of $\CA$ on $\alpha$.
(We assume here that input alphabet and output alphabet coincide; it is obvious how to cover the case of different alphabets.)

\begin{remark}
    Our model of transducer corresponds to Moore automata in classical automata theory, in which transitions involve changes of states (here: configurations) induced by inputs and where outputs are determined by states (here: configurations).
    Other models, for example closer to Mealy automata, are also conceivable and could be used in our main result.
    We do not enter here a discussion on such variants, their comparison in expressive power, and their use in the solution of infinite games.
\end{remark}

\begin{remark}
    The format of transition formulas and output formulas in $\NN$-MSO-automata and $\NN$-MSO-transducers can be modified in a way that the indexing of formulas by states $p,q \in Q$ is avoided:
    One may use the structure $Q \times \CN_\mathrm{succ}$ for configurations rather than $\CN_\mathrm{succ}$ for memory contents and then apply an indexing  on the level of elements of the domain $Q\times \NN$ to access the control states and the memory contents.
    We decided for an indexing on the level of formulas  and to use the standard structure $\CN_\mathrm{succ}$ for the interpretation of formulas.
\end{remark}

We end this subsection by giving a ``machine characterization'' of $\NN$-MSO-automata in which the transition formulas and the output formulas are replaced by suitable automata, expanding on the Remark \ref{automatic} above that the relations defined by transition formulas are automatic.
For this purpose we give the input letters, i.e.,
the natural numbers, a more concrete representation.
It will be convenient (also for a later use in the context of  pushdown systems)
to represent the number $m$ by the $\omega$-word $\bot 1^m \# \# \ldots$ which we imagine written as a column with $\bot$ at the bottom.
The input word  $ 4 \ 2 \ 0 \ 4 \ 3$, for instance, is represented by the sequence of columns in Figure~\ref{gridfigure}.

\begin{figure}[ht]
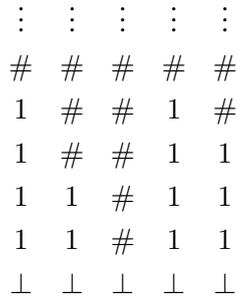

\vspace*{-1mm}
    \[
        \begin{matrix}
            \vdots & \vdots & \vdots & \vdots & \vdots & \\
            \#   & \#   & \#   & \#   & \#   & \\
            1   & \#   & \#   & 1   & \#   & \\
            1  & \#   & \#  & 1   & 1   & \\
            1  & 1   & \#   & 1   & 1   & \\
            1   & 1   & \#  & 1   & 1   & \\
            \bot     & \bot     & \bot     & \bot     & \bot     & \
        \end{matrix} \vspace*{-4mm}
    \]
    \caption{\label{gridfigure}Representation of the word $4 \ 2 \ 0 \ 4 \ 3$ as a sequence of columns.}
\end{figure}

A run of an $\NN$-MSO-automaton $\CA$ corresponds to an enrichment of this grid structure where to the left and to the right of each column information on a configuration is added; then one requires that the configuration at the right of a column coincides with the configuration to the left of the subsequent column.
A  transition from configuration $(p,i)$ via letter $m$ to configuration $(q,j)$ is described using two auxiliary pebbles,
an ``in-pebble'' colored $p$ at the position of the $i$-th letter $1$  left of the column $\bot 1^m \# \# \ldots$ and an ``out-pebble'' colored $q$ at the position of the $j$-th letter $1$ right of the column.
We can  ``implement'' transitions according to a set of formulas $(\varphi_{pq}(x,z,y))_{p, q \in Q}$ by an automaton $\CC$ that computes a target configuration $(q,j)$ from a given configuration $(p,i)$ and an input number~$m$:
This ``transition automaton'' $\CC$ works in a two-way manner, starting at the bottom of a column $\bot 1^m \# \# \ldots$ that is equipped with an in-pebble, say colored $p$ at position $i$.
The automaton can place the out-pebble at any position it reaches (thereby removing it from the previous position, with the initial position being $0$), and at each position can inspect the respective column letter (from $\{\bot, 1, \#\}$) and the in- and the out-pebble with their colors as far as they are  there.
The computation ends in a designated halting state, and the color $q$ and the position $j$ of the out-pebble at its last placement give the new configuration.
In~\cite{BLT17} it is shown that the transition formulas of a nondeterministic or deterministic $\NN$-MSO-automaton can be substituted with a nondeterministic or deterministic transition automaton~$\CC$, respectively, and that conversely the definition of transitions in terms of such an automaton $\CC$ can be transformed into a corresponding family $(\varphi_{pq}(x,z,y))_{p, q \in Q}$ of MSO-formulas.

In order to extend the model to cover also $\NN$-MSO-transducers, we add to a transition automaton $\CC$ an ``output automaton'' $\CD$ that places a further ``output-pebble'' on some position after the deterministic transition automaton $\CC$ has halted on a given column.
Again $\CD$ proceeds in a two-way mode, can inspect the current position of the out-pebble and can place the output-pebble (thereby removing it from a previous position, with the initial position being $0$).
The automaton $\CD$ can halt by entering a designated halting state, and the last position of the output-pebble then defines the current output.

\subsection{Some properties}

We give some comments on the properties of $\NN$-MSO-automata, using an informal exposition and mainly focusing on languages of finite words (to allow for direct comparisons with existing models).
For a more detailed treatment we refer to~\cite{BLT17}.

As shown in~\cite{BLT17}, the $\NN$-MSO-automata are different in expressive power from other models considered in the literature; let us mention the finite memory automata of Kaminski and Francez~\cite{KF},
the automata on data words of Boja\'{n}czyk, David, Muscholl, Schwentick, and S\'{e}goufin~\cite{BDMS}, the ordered data automata of Tan~\cite{Tan12}, the G-automata of  Boja\'{n}czyk, Klin and Lasota~\cite{BKS14}, and the single-use automata of Boja\'{n}czyk and Stefa\'{n}ski~\cite{BS20}.
A language that is recognizable by $\NN$-MSO-automata but not by any of the mentioned ones is the set  $L_1$ of words of the form $n (n+1) (n+2) \ldots (n+\ell)$ considered in the previous subsection.

Not surprisingly, deterministic $\NN$-MSO-automata are strictly weaker in expressive power than nondeterministic ones.
A simple example is the language $L_2$ (over $\NN$) of words in which some number occurs twice;
by nondeterminism the position of its first occurrence is guessed, and its value is propagated until the second occurrence.
(In the previous subsection we considered an analogous $\omega$-language $L_2'$ containing the sequences in which some number is repeated infinitely often.)
It is intuitively clear that deterministic $\NN$-MSO-automata fail to recognize $L_2$ (as well as $L_2'$) since a possibility to store several candidates of natural numbers during a computation is lacking (for a detailed proof see~\cite{CST15}).
This weakness does not arise for the language that contains the words over $\NN$ in which the \emph{first} number reappears some time later (or, in an $\omega$-word, infinitely often).
Repetitions of values play a role in the study of ``data words''; some of the  above-mentioned automaton models are designed to handle multiple occurrences of values (and thus to recognize $L_2$).

Let us mention that the decision problem whether a language recognized by a nondeterministic $\NN$-MSO-automaton is in fact recognizable by a deterministic one is undecidable;
similarly this problem is undecidable when parity $\NN$-MSO-automata over infinite words are considered.
This follows from the corresponding result on 1-counter automata over finite alphabets (cf.\ Theorem~12.4.3 in~\cite{Bru18}).

The closure properties of $\NN$-MSO-automata in their deterministic and nondeterministic form resemble those of pushdown automata (and a connection to pushdown automata will be used later in this paper).
For example, closure under intersection fails;
as an example one can use the intersection of the languages $L_3$ and $L_3'$ (over $\NN$) where $L_3$ consists of the words in which the first and the third number coincide and $L_3'$ consists of the words in which the second and the fourth number coincide (cf.~\cite{BLT17}).

We mention an important positive result on $\NN$-MSO-automata (again from~\cite{BLT17}):

\begin{proposition}\label{nonempty}
    The non-emptiness problem for $\NN$-MSO-automata over finite words is decidable.
\end{proposition}

For the proof it suffices, given an $\NN$-MSO-automaton $\CA$, to construct an  MSO-sentence $\varphi_\CA$ which says in $\CN_\mathrm{succ}$ that $\CA$ accepts at least one word;
then one can apply decidability of the MSO-theory of~$\CN_\mathrm{succ}$.
To that end, we present a formula $\psi_{q,q'}(x,x')$ which is true in the structure $\CN_\mathrm{succ}$ iff there exists a word on which $\CA$, starting in configuration $(q, x)$, can reach configuration $(q',x')$;
then we set $\varphi_\CA \coloneqq  \bigvee_{q \in F} \exists y \psi_{q_0,q}(0,y)$, where $q_0$ is the initial state and $F$ is the set of accepting states of $\CA$.

If the set of states of $\CA$ is $\{q_0,\ldots, q_n\}$ we introduce corresponding set variables $X_{q_k}$.
Let $\theta(X_{q_0}, \ldots, X_{q_n})$ be the formula expressing that the set tuple $(X_{q_0}, \ldots, X_{q_n})$ is  closed under the
\mbox{transition} relation defined by the family of transition formulas $\varphi_{p,q}$ of $\CA$:
\[
\theta(X_{q_0},\ldots,X_{q_n})\coloneqq\forall x \forall y (\bigwedge_{p,q \in Q} (X_{p}(x) \wedge \exists z\, \varphi_{p,q}(x,z, y)) \rightarrow X_{q}(y)).
\]
Now we  set
\[
\psi_{q,q'}(x,x')\coloneqq\forall X_{q_1}\ldots\forall X_{q_n} ((X_q(x) \wedge \theta(X_{q_1},\ldots, X_{q_n}))\rightarrow X_{q'}(x')).
\]
Note that the membership problem (whether an $\NN$-MSO-automaton $\CA$ accepts a given word $w \in \NN^*$) is solvable by enhancing $\CA$ with a test whether the sequence of scanned letters matches $w$ and then checking for non-emptiness.

For deterministic $\NN$-MSO-automata (trivially closed under complement), the decidability of the non-emptiness problem extends to the universality problem.
However, for nondeterministic $\NN$-MSO-automata the universality problem is undecidable (cf.~\cite{CST15}).
The decidability of the non-emptiness problem also holds for parity automata over infinite words; it is easy to extend the idea of the proof above, using induction on the number of colors.

\subsection{Some variants of \texorpdfstring{$\NN$}{N}-MSO-automata}\label{variants}

There are many ways to modify the type of input and the memory structure of $\NN$-MSO-automata, as well as the logic used for transition formulas.
We present some of such variants, focusing on the question whether the non-emptiness problem is decidable and thus whether there is some algorithmic potential.
\begin{enumerate}
    \item \emph{Presburger arithmetic.}
        Here one uses Presburger arithmetic for the definition of transitions, again with natural numbers as memory contents but now with first-order transition formulas over $(\NN, +, 0)$.
        For this model the non-emptiness problem is undecidable, since the computations of 2-counter automata (or 2-register machines) can be simulated by automaton runs, even discarding input letters.
        A pair $(m, n)$ of numbers is coded by the number $k = 2^m 3^n$; then, for example, increment of $m$ to $m+1$ (and of $n$ to $n+1$, respectively)
        is simulated by passing from $k$ to $k+k$ (from $k$ to $k+k+k$, respectively).

    \item \emph{Words as input letters.}
        Next we consider the case that the input letters are no more natural numbers but words over a finite alphabet $\Sigma$.
        (The case where $\Sigma = \{1\}$ amounts to the original model, where the alphabet is $\NN$.)
        A finite word over the alphabet $\Sigma^*$ is a finite sequence $u_0 u_1 \ldots u_\ell$ of ``$\Sigma^*$-letters'' $u_k \in \Sigma^*$;
        it should not be confused with the concatenation of the $u_k$.
        Similarly, an $\omega$-word over $\Sigma^*$ is a sequence $u_0 u_1 u_2 \ldots$ with $u_k \in \Sigma^*$.
        For the logical description of transitions in the structure $\CN_\mathrm{succ}$,
        a $\Sigma^*$-letter $u = b_1 \ldots b_\ell$ is represented by a $|\Sigma|$-tuple of subsets of $\NN$, namely  $(P_a)_{a \in \Sigma}$ where we have $m \in P_a$ iff $b_m = a$.
        In the transition formulas one uses set variables $Z_a$ to refer to these sets.
        A transition formula is then of the form $\varphi_{pq}(x, \overline{Z}, y)$, where $\overline{Z} = (Z_a)_{a \in \Sigma}$ is a $|\Sigma|$-tuple of set variables; similarly one obtains output formulas $\psi_q(x, \overline{Z})$.
        For the resulting automaton model, all properties that were mentioned for $\NN$-MSO-automata hold again with the obvious modifications.
        It should be mentioned that also infinite words can be taken as input symbols while still keeping decidability of the non-emptiness problem.

        The theory of $\NN$-MSO-automata over alphabets $\Sigma^*$ and the solution of infinite games over alphabets $\Sigma^*$ is treated in detail in the first author's dissertation~\cite{Bru18}.
        In the present paper we consider the special case of the alphabet $\NN$ (and thus games with winning conditions in the Baire space) since this case spares some technicalities and in the literature on infinite games the Baire space usually serves as the master example when infinite alphabets are treated.

    \item\label{word-memory} \emph{Words as memory contents.}
        Next we consider the case where words $u \in \Sigma^*$ (with $|\Sigma| > 1$) not only occur as input letters but also as memory contents.
        Then the transition formulas could take the form $\varphi_{pq}(\overline{X}, \overline{Z}, \overline{Y})$ with variables $\overline{X}, \overline{Z}, \overline{Y}$ for set tuples.
        By the equivalence between MSO-logic and finite automata,
        the transition relations $T_{pq} \subseteq (\Sigma^*)^3$ are automatic (as it was mentioned in Remark \ref{automatic} for the case of a singleton alphabet).
        In this framework, the update of Turing machine configurations is directly representable  (even using input-free automata) and one obtains  the undecidability of the non-emptiness problem.
        This effect is avoided when one takes a weaker format of the transition formulas, using the infinite $|\Sigma|$-branching tree $\CT_{\Sigma^*}$ whose nodes are the words over $\Sigma$,
        and which is equipped with the successor functions $u \mapsto ua$ for $ a \in \Sigma$.
        Using transition formulas  $\varphi_{pq}(x,z,y)$ interpreted in this tree, one can show that the non-emptiness problem is decidable by copying the proof of Proposition~\ref{nonempty} and invoking decidability of the MSO-theory of $\CT_{\Sigma^*}$.
        However, with the corresponding definition of transducers a transfer of the main result presented below is not possible due to the fact  that over trees the definable selection of an output tree matching a given MSO-property fails (\cite{carayol_choice_2010}); see Remark~\ref{numberchoice}.

    \item\label{tree-letters} \emph{Finite or infinite trees over ranked alphabets as input letters.}
        This case can be handled in analogy to the case where the input letters are finite words or infinite words, respectively;
        a transition formula has the format $\varphi_{pq}(x,\overline{Z},y)$ where $\overline{Z}$ stands for the labelling of a tree and $x,y$ for two positions (i.e., nodes) in or beyond the domain of this tree.
        Decidability of the non-emptiness problem now follows as in Proposition \ref{nonempty}, again by applying Rabin's Tree Theorem~\cite{Rab69} that the MSO-theory of the $k$-ary infinite tree is decidable (where $k$ is the size of the ranked alphabet under consideration).
        Again, as mentioned above,
        a principal difficulty arises when corresponding transducers for  trees are envisaged for a solution of the synthesis problem since an MSO-definable selection of an output tree is not possible (\cite{carayol_choice_2010}).
\end{enumerate}

\subsection{Infinite games and the synthesis problem}

Using parity $\NN$-MSO-automata we now define corresponding infinite games over the Baire space.
Such a game (also called a Gale-Stewart game, cf.~\cite{GaleStewart53}) is played between two players called I and II, who choose natural numbers in alternation:
Player~I starts by choosing a number $m_0$, then Player~II responds by choosing a number $n_0$,
then Player~I picks a number $m_1$, Player~II a number $n_1$, and so on ad infinitum.
The resulting sequence $m_0 n_0 m_1 n_1 m_2 \ldots$ is a play of the game.
The winning condition is defined by an $\omega$-language $L \subseteq \NN^\omega$:
A play $\alpha$ is declared to be won by Player~II iff $\alpha$ belongs to $L$, otherwise Player~I wins the play $\alpha$.
By $\Gamma(L)$ we denote the game with winning condition $L$.
A strategy for Player~II is given by a function $\sigma\colon \NN^* \rightarrow \NN$; it fixes the choice of Player~II upon the choices $m_0, \ldots, m_k$ of Player~I to be $n_k = \sigma(m_0 \ldots m_k)$.
For a sequence $\alpha_{\pOne/} = m_0 m_1 \ldots $ of numbers chosen by Player~I, the function $\sigma$  induces a sequence $\alpha_{\pTwo/} = n_0 n_1 \ldots$ and thus the play $\interleave{\alpha_{\pOne/}}{\alpha_{\pTwo/}} = m_0 n_0 m_1 n_1 \ldots$.
The strategy $\sigma$ is called a winning strategy for Player~II in $\Gamma(L)$ if for every sequence $\alpha_{\pOne/}$ the induced play $\interleave{\alpha_{\pOne/}}{\alpha_{\pTwo/}}$ belongs to $L$.
If a winning strategy exists for Player~II in $\Gamma(L)$, we say that Player~II wins this game.
Correspondingly one defines strategies and winning strategies for Player~I.
The game $\Gamma(L)$ is called \emph{determined} if one of the two players has a winning strategy.
It is well known that the statement ``For each $L \subseteq \NN^\omega$ the game $\Gamma(L)$ is determined'' is a strong set-theoretical hypothesis, incompatible with the axiom of choice; however, for Borel sets $L\subseteq \NN^\omega$, Martin's Theorem~\cite{Mar75} guarantees determinacy of $\Gamma(L)$.

In the sequel we study games defined in terms of $\omega$-languages recognized by deterministic  parity $\NN$-MSO-automata.
Due to the parity acceptance condition, such games belong to the Boolean closure of $\Pi_2$-sets in the Borel hierarchy and thus are determined.
We are interested here in two algorithmic problems, namely to decide, given a deterministic parity $\NN$-MSO-automaton $\CA$, who of the two players wins the game $\Gamma(L_\CA)$,
and to provide a concrete winning strategy for the winner.
This is the content of Church's Problem~\cite{Chu63} which was originally stated for sequences over finite alphabets, asking whether for a given condition on sequences $\interleave{\alpha_{\pOne/}}{\alpha_{\pTwo/}}$ a ``solution'' in terms of a letter-by-letter computable transformation $\sigma\colon \alpha_{\pOne/} \mapsto \alpha_{\pTwo/}$ exists and in this case to construct such a solution.
In~\cite{McN65}, McNaughton introduced the game-theoretic formulation of this problem, giving it a symmetric shape in the two players, and thus asking for an explicit winning strategy for Player~II if such a strategy for Player~I does not exist.
It should be noted that Church mentioned conditions presented in a logical formalism.
He spoke of ``systems of recursive arithmetic''; later the dominant logic in this context was MSO-logic over $\CN_\mathrm{succ}$,
since the B\"uchi-Landweber-Theorem~\cite{BL69} provided a solution of Church's Problem for this logic.
In the present paper we consider the alphabet $\NN$ and start with a condition defined in terms of deterministic parity $\NN$-MSO-automata.

\section{Synthesis}

\subsection{Main result and proof strategy}

\begin{theorem}\label{maintheorem}
    For a game $\Gamma(L_\CA)$ where $\CA$ is a deterministic parity $\NN$-MSO-automaton, the winner can be computed, and an $\NN$-MSO-transducer realizing a winning strategy for the winner can be constructed.
\end{theorem}

In the proof we use (modifications of) results on definability of winning strategies in pushdown games.
Thus we pass from arithmetic, as it appears in $\NN$-MSO-automata and where numbers are the elementary objects,
to formal languages, where we deal with words rather than numbers and handle single digits (in our case in pushdown systems).
Definable strategies in this context are then translated back to arithmetic as it is required in our format of transducers.

At several places we rely on (or modify) constructions from the literature.
In particular, we build on ideas of Kupferman, Piterman, and Vardi on alternating automata~\cite{KPV10, piterman_micro-macro_2003}
and on work by Cachat~\cite{cachat_uniform_2003} on parity games over prefix-recognizable graphs,
involving a ``reduction'' of prefix-recognizable graphs to pushdown graphs that goes back to Caucal (see Proposition 4.2 of~\cite{Caucal03}).

As a first step, we extract from the transition graph of $\CA$ a game graph $G(\CA)$;
this graph just captures the progression of plays in terms of moves from one configuration to another.
So in $G(\CA)$ the transition labels are no more present.
The alternation of moves between the two players is captured by doubling the set of control states (so that one obtains ``I-states'', written  $p_\pOne/$, and ``II-states'', written $p_\pTwo/$).
One obtains a \emph{prefix-recognizable graph}; its set of vertices is a regular language of words over a finite alphabet (coding the configurations)
and its set of edges is MSO-definable in a sense explained below.
Vertices may have infinite out-degree.
A play is now an infinite path through the graph (possibly with repetitions of vertices), and the parity condition inherited from $\CA$ defines the winning condition.

The next step is the transformation of the game graph $G(\CA)$ into a game graph $H(\CA)$ where each vertex has bounded finite out-degree.
In this new game graph, an edge of $G(\CA)$ is simulated by a finite sequence of ``micro-steps''.
The two players do no more alternate between moves but a player can perform a sequence of moves rather than a single move.
The new graph $H(\CA)$ has the format of a \emph{pushdown game graph}:
The vertices consist of a stack content and a control state, and edges allow to change the top of the stack according to finitely many pushdown rules.
Using the fact that in a parity game the winner has a \emph{memoryless winning strategy},
i.e., a winning strategy where the choice of a transition only depends on the current vertex,
we can code such a strategy $\sigma$ by an infinite labelled tree with labels from a \emph{finite} alphabet.

In the third step we show that the labelled trees defining winning strategies for a given player constitute a \emph{regular} tree language (at first represented by a universal two-way parity tree automaton), which is nonempty iff the given player has a winning strategy.
In this case, applying a basic fact of the theory of regular tree languages,
we can construct a \emph{regular tree} belonging to the language, representing a ``regular'' winning strategy;
the labels of the nodes of such a tree can be produced by a deterministic finite automaton.

The last steps consist in transforming a regular winning strategy over $H(\CA)$  back to the original game graph $G(\CA)$ and then to reestablish the reference to input letters of $\CA$ that were cancelled in the formation of $G(\CA)$ (i.e., the letters chosen by the players when building up a play).
This will enable us to construct the desired $\NN$-MSO-transducer realizing a winning strategy for the winner of the game~$\Gamma(L_\CA)$.

\subsection{Some preliminaries}\label{preliminaries}

\subsubsection{MSO-logic}\label{msologic}

We assume that the reader is acquainted with the fundamental results connecting MSO-logic over finite words to finite automata, over infinite words to B\"uchi automata (and parity automata),
and over infinite labelled binary trees to finite automata over infinite trees (e.g., parity tree automata).
For example,
we shall switch from automata to MSO-logic (or conversely) without giving further details when this is convenient; in the present paper we do not enter questions of complexity.
We now fix some terminology and recall some facts.

As a variant of the infinite binary tree we shall use, for any $k \geq 1$,
the $k$-ary branching tree.
Its vertices can be identified with the words over an alphabet $\AAA =\{a_1, \ldots, a_k\}$;
we write $t_{\AAA^*}$ rather than $\AAA^*$ when emphasizing the view of $\AAA^*$ as a tree.
Together with the successor functions $\mathrm{succ}_{a_i} : w \mapsto w a_i$ we obtain the tree structure
$\CT_{\AAA^*} \coloneqq (t_{\AAA^*}, \mathrm{succ}_{a_1}, \ldots, \mathrm{succ}_{a_k})$.
By Rabin's Tree Theorem, the MSO-theory of $\CT_{\AAA^*}$ is decidable.

In the proof of our main result, automata on infinite labelled trees arise.
We mention some results that we need.
If $t$ is a labelling of $t_{\AAA^*}$ by letters in some alphabet $\BBB$, we write $t(w)$ for the label at position $w \in \AAA^*$.
As automata on infinite labelled trees we use nondeterministic parity tree automata.
Such an automaton is equipped with a coloring $c\colon Q \rightarrow \{0, \ldots, d\}$, where $Q$ is the set of states, and a tree $t$ is accepted if a run on $t$ exists where on each path the maximal color occurring infinitely often is even.
In the proof of our main result we consider variants of this basic model, namely universal parity tree automata, moreover in the two-way version,
for which we use the fact that these are expressively equivalent to nondeterministic parity tree automata (cf.~\cite{KPV10}).

A $\BBB$-labelled tree $t$, labelling $t_{\AAA^*}$, is called \emph{regular} if there is a deterministic  finite automaton with output that produces label $t(w)$ after reading the input word $w \in \AAA^*$.
To construct a regular tree means to construct a corresponding generating automaton.
We use the following fundamental fact which underlies the decidability proof for the MSO-theory of $\CT_{\AAA^*}$ (cf., e.g.,~\cite{Tho97}):
\begin{theorem}\label{rabinbasis}
    It is decidable whether the tree language recognized by a given nondeterministic parity tree automaton $\CA$ is nonempty, and in that case a regular tree accepted by $\CA$ can be constructed.
\end{theorem}
Tree structures $\CT_{\AAA^*}$ (for suitable alphabets $\AAA$) play a major role in this paper since we shall consider game graphs (formally introduced in the following subsection) that can be defined in such tree structures by MSO-formulas.
If a graph $G = (V,E)$ is defined in $\CT_{\AAA^*}$ by an MSO-formula $\varphi(x)$ and an MSO-formula $\psi(x,y)$ in the sense that
\[
    V = \{u \in \AAA^* \mid  \CT_{\AAA^*} \models \varphi[u]\}, \ \ E = \{(u,v) \in \AAA^* \times \AAA^* \mid \CT_{\AAA^*} \models \psi[u,v]\},
\]
one says that $G$ is MSO-interpretable in $\CT_{\AAA^*}$ or that $G$ is \emph{prefix-recognizable}.
Analogously one defines prefix-recognizable graphs of slightly different format, e.g., when the vertex set $V$ falls into two disjoint sets $V_1, V_2$.
In the literature the prefix-recognizable graphs are only required to be isomorphic to a graph MSO-defined in a tree structure $\CT_{\AAA^*}$; in the present paper we use the more restrictive definition given above, where the prefix-recognizable graph itself must be MSO-defined in $\CT_{\AAA^*}$.

We shall need a characterization of the edge relation of prefix-recognizable graphs assuming that the edge relation is defined by an MSO-formula $\psi(x,y)$  in $\CT_{\AAA^*}$ (see, e.g.,~\cite{BCL07} or Chapter 15 of~\cite{GrThWi}):
\begin{proposition}\label{UVWproposition}
For each MSO-formula $\psi(x,y)$ (interpreted in $\CT_{\AAA^*}$) there are regular languages $U_i, V_i, W_i \subseteq \AAA^*$ ($i = 1, \ldots, k$)
such that
\[
    \CT_{\AAA^*} \models \psi[w_1, w_2]   \ \ \mbox{iff}   \ \ \bigvee_{i = 1, \ldots, k} \exists u \in U_i \ \exists v \in V_i \ \exists w \in W_i \  (wu = w_1 \wedge wv = w_2).
\]
Conversely, for each family of regular languages $U_i, V_i, W_i \subseteq \AAA^*$, there is an MSO-formula $\psi(x,y)$ with that property.
\end{proposition}
So an edge between vertices $w_1, w_2$ can be decomposed into a triple $(u,v,w)$ of words, using a ``rule'' $(U_i, V_i, W_i)$ as follows:
There is a prefix (a common ancestor) $w \in W_i$ of $w_1$ and $w_2$ from which one can reach $w_1$ by attaching a word $u \in U_i$ and from which one can reach $w_2$ by attaching a word $v \in V_i$.

In the sequel it will be necessary to connect statements on MSO-definability in trees $\CT_{\AAA^*}$ for different alphabets $\AAA$.
In particular, we use the following fact that follows from Proposition~\ref{UVWproposition}:

\begin{proposition}\label{alphabetrestriction}
Let $\BBB$ and $\CCC$ be finite alphabets with $\BBB \subseteq \CCC$ and let $\varphi(x,y)$ be an MSO-formula such that
if $\CT_{\CCC^*} \models \varphi[u,v]$ then $u,v \in \BBB^*$.
Then there is an MSO-formula $\psi(x,y)$ such that
$\CT_{\BBB^*} \models \psi[u,v]$ iff $\CT_{\CCC^*} \models \varphi[u,v]$,
for all $u,v \in \BBB^*$.
\end{proposition}

Finally, we recall the fact that MSO-logic allows to express the transitive closure of binary MSO-definable relations.
Already the proof of Proposition \ref{nonempty} relied on this feature of MSO-logic.
Let us consider an MSO-formula $\psi(y_1,y_2)$ defining the edge relation of a graph.
We write $\psi^*(x,y)$ to express that there is a sequence of vertices $x_0, \ldots , x_k$ leading from $x$ to $y$ via edges defined by $\psi$, which means that $x_0 = x$, $x_k = y$, and $\psi(x_i, x_{i+1})$ holds for $i = 0, \ldots , k-1$; this condition is expressible by the following MSO-formula:
\[
    \psi^*(x,y) \ \coloneqq \ \forall X (X(x) \wedge \forall y_1 \forall y_2 ((X(y_1) \wedge \psi(y_1, y_2)) \rightarrow X(y_2)) \rightarrow X(y)).
\]

\subsubsection{Games over graphs}

In the sequel it is useful to work with the concept of game graph and related notions and results; we recall them briefly.
A \emph{game graph} has the form $G = (V_1, V_2, E)$; its vertex set $V$ is composed of two disjoint subsets $V_1, V_2$.
In strictly turn-based games, edges lead from $V_1$-vertices to $V_2$-vertices and conversely, i.e., $E \subseteq (V_1 \times V_2) \cup (V_2 \times V_1)$;
the idea is that if a play has reached vertex $u \in V_1$ then Player~I has to choose an edge from $u$ to some vertex $v \in V_2$, from where then Player~II has to choose an edge to a vertex in $V_1$, etc.
When the parity condition is used, the game graph is presented as $G = (V_1, V_2, E, c)$
where $c\colon  V_1 \cup V_2 \rightarrow \{0, \ldots , d\}$ for some natural number $d$.
At a later point we shall also consider games that are not strictly turn-based; then edges may also lead from $V_1$-vertices  to $V_1$-vertices,
analogously for $V_2$.
In order to avoid deadlocks we assume that each vertex has an out-degree $\geq 1$.

 \medskip
A \emph{play from $u_0$} in a game graph $G$ is a path through $G$ that starts in vertex $u_0$, proceeding along edges ad infinitum,
where Player~I and Player~II choose edges from vertices in $V_1$ and $V_2$, respectively.
To determine the winner of a play,
a winning condition for Player~II is  invoked; in our case we refer to a coloring $c\colon V_1 \cup V_2  \rightarrow \{0, \ldots, d\}$ and use the parity condition (requiring that the maximal color occurring infinitely often in the play is even).
A \emph{strategy from $u_0$} for Player~II, say, is a function mapping finite play prefixes starting in $u_0$ to vertices, more precisely mapping a play prefix starting in $u_0$ and ending in $u \in V_2$ to  a vertex $v$ such that $(u, v) \in E$.
If the vertex selected by the function only depends on the last vertex of the given play prefix, i.e., the ``current vertex'', the strategy is called \emph{memoryless}.
A \emph{winning strategy from $u_0$} for Player~II, say, is a strategy $\sigma$ which ensures satisfaction of the winning condition for each play starting in $u_0$ that is compatible with  $\sigma$,
i.e., for arbitrary choices of moves by the opponent Player~I.
The \emph{winning region} of a player is the set of vertices $u_0$ such that there is a winning strategy from $u_0$ for this player.
A player has a \emph{uniform} memoryless strategy on her winning region $R$ if the value for a vertex $u \in R$ can be chosen independently of the start vertices $u_0 \in R $ of plays.

 \medskip
For parity games, i.e., games with a parity winning condition,
we use the following fundamental result (\cite{EmersonJutla91}, see also the survey~\cite{Tho97}):

\begin{theorem}
    In a parity game over $G = (V_1, V_2, E, c)$, the winning regions of the two players form a partition of $V_1 \cup V_2$, and each player has a uniform memoryless winning strategy on her winning region.
\end{theorem}

Following the explanations of the previous subsection, we define a parity game graph
$G = (V_1, V_2, E, c)$ with $c\colon  V_1 \cup V_2 \rightarrow \{0, \ldots , d\}$
to be \emph{prefix-recognizable} if $G$ is defined in some tree structure $\CT_{\AAA^*}$
by MSO-formulas $\varphi_1(x), \varphi_2(x), \psi(x,y), \delta_0(x), \dotsc, \delta_d(x)$
in the sense that
\begin{itemize}
\itemsep=0.9pt
	\item
        $V_1 = \{u \in \AAA^* \mid  \CT_{\AAA^*} \models \varphi_1[u]\}$,~
        $V_2 = \{u \in \AAA^* \mid  \CT_{\AAA^*} \models \varphi_2[u]\}$,
	\item $E = \{(u,v) \in \AAA^* \times \AAA^* \mid \CT_{\AAA^*} \models \psi[u,v]\}$, and
	\item $c(u) = i$ iff $\CT_{\AAA^*} \models \delta_i[u]$, for all $i \in \{0,\dotsc,d\}$.
\end{itemize}

\subsection{Building a prefix-recognizable game graph \texorpdfstring{$G(\CA)$}{G(A)}}

Let  $\CA = (Q, q_0, (\varphi_{pq}(x, {z},y))_{p,q \in Q}, c)$ be a deterministic parity $\NN$-MSO-automaton over the alphabet~$\NN$.
In order to represent the alternation of moves between players I and II, the state set $Q$ is doubled into two disjoint sets $Q_{\pOne/}, Q_{\pTwo/}$.
So for each $q \in Q$ we have two corresponding states $q_{\pOne/}$ (from where Player~I moves) and $q_{\pTwo/}$ (from where Player~II moves).

From $\CA$ we build a game graph whose vertices are finite words over the alphabet $\{1\} \cup Q_{\pOne/} \cup Q_{\pTwo/}$, representing configurations of $\CA$, now with control states from  $Q_{\pOne/} \cup Q_{\pTwo/}$.
For example, the word $1^i q_\pTwo/$ represents the configuration $(q,i)$ with Player~II to move.
An edge of $G(\CA)$ will represent a transition of $\CA$, discarding the respective input letter from $\NN$, and taking into account that from a I-state a II-state (or from a II-state a I-state) will be reached.

\begin{definition}
    For a deterministic parity $\NN$-MSO-automaton $\CA = (Q, q_0, (\varphi_{pq}(x, {z},y))_{p,q \in Q}, c)$
    we denote by $G(\CA)$ the game graph $(V_1, V_2, E, c')$ where
    \begin{itemize}
    \itemsep=0.9pt
        \item
            $V_1 = \{1^k p  \mid  k \geq 0, p \in Q_\pOne/ \}$,~$V_2 = \{1^k p  \mid  k \geq 0, p \in Q_\pTwo/ \}$\\
            with $Q_\pOne/ = \{q_\pOne/ \mid q \in Q\}$, $Q_\pTwo/ = \{q_\pTwo/ \mid q \in Q\}$,
        \item
            the set $E$ contains the edges of the form $(1^i p_\pOne/, 1^j q_\pTwo/)$ and $(1^i p_\pTwo/, 1^j q_\pOne/)$ corresponding to  transitions of $\CA$, i.e., such that
            $\CN_\mathrm{succ}  \models \exists z\, \varphi_{pq}(x,z,y) \, [i, j]$, for $p,q \in Q$,
        \item
            the coloring $c'$ is inherited from $c$, i.e., for any vertex $u = 1^i q_\pOne/$ or $u = 1^i q_\pTwo/$  we have $c'(u) = c(q)$.
    \end{itemize}
\end{definition}
Each play  in the game $\Gamma(L_\CA)$ induces a configuration sequence of $\CA$, which in turn corresponds to a path through the game graph $G(\CA)$:
For the configuration sequence $(q_0,0), (p_1, i_1), (p_2, i_2), \ldots$ the corresponding path in $G(\CA)$ is the vertex sequence
$(q_0)_\pOne/, 1^{i_1} (p_1)_\pTwo/, 1^{i_2} (p_2)_\pOne/, \ldots$.
Conversely, each infinite path through $G(\CA)$ will correspond to a  configuration sequence according to the $\CA$-transitions, where however the start vertex may be any configuration (possibly different from $(q_0,0)$), and the states of $Q_{\pOne/}$, $Q_{\pTwo/}$ are projected to $Q$.
In the edges of $G(\CA)$ the numbers serving as input entries of the $\CA$-transitions are suppressed; only the corresponding changes of configurations are recorded.
The definition of $c'$ from $c$ ensures that the winning condition over $G(\CA)$ depends only on the control states of configurations (as it does over $\CA$) and hence is satisfied for a play over $G(\CA)$ iff it is satisfied for the corresponding run of $\CA$.

\eject

The definition of $G(\CA)$ from $\CA$ immediately yields the following proposition.
\begin{proposition}
    Let $\CA$ be a deterministic parity $\NN$-MSO-automaton.
    Player~II wins the game $\Gamma(L_\CA)$ iff the vertex $(q_0)_\pOne/$ of $G(\CA)$ (representing the initial configuration $(q_0,0)$ of $\CA$)
    belongs to the winning region of Player~II in the parity game over $G(\CA)$.
    Analogously, Player~I wins the game $\Gamma(L_\CA)$ iff the vertex $(q_0)_\pOne/$ belongs to the winning region of Player~I in the parity game over $G(\CA)$.
\end{proposition}

\begin{proposition}\label{prefixrec}
    The parity game graph $G(\CA) = (V_1, V_2, E, c')$ is prefix-recognizable.
\end{proposition}
\begin{proof}
    The vertices of $G(\CA)$ are words over the alphabet $\BBB \coloneqq \{1\} \cup Q_{\pOne/} \cup Q_{\pTwo/}$ and thus nodes of the $(2 |Q| +1)$-branching tree $\CT_{\BBB^*}$.
    The two vertex sets $V_1 = 1^* Q_{\pOne/}$ and $V_2 = 1^* Q_{\pTwo/}$ are regular languages and hence MSO-definable in $\CT_{\BBB^*}$.
    The same is true for the preimages of the colors under the coloring function $c'$, which can be written as $1^* c^{-1}(0), \dotsc, 1^* c^{-1}(d)$, where $c$ is the coloring function of~$\CA$.
    Now we show that the edge relation $E$ is MSO-definable in $\CT_{\BBB^*}$.
    First we find an MSO-formula $\varphi'_{pq}(x,y)$ such that
    \[
        \CN_\mathrm{succ}   \models \exists z \, \varphi_{pq} \, [i,j] \mbox{  \ \ iff \ \ }  \CT_{\BBB^*} \models \varphi'_{pq} \, [1^i,  1^j].
    \]
    The binary relation defined by $\exists z\,\varphi_{pq}(x,z,y) $ is automatic (cf.~Remark \ref{automatic}) and thus definable by a finite automaton $\CA_{pq}$ scanning pairs $(1^i, 1^j)$ of words in parallel (and terminating when the longer word ends).
    The translation of $\CA_{pq}$ to MSO-logic over $\CT_{\BBB^*}$ yields the desired formula $\varphi'_{pq}(x,y)$.
    The edge relation of $G(\CA)$ is then defined by the following formula:
    \[
        \psi(x,y) \coloneqq \bigvee_{(p,q) \in (Q_{\pOne/} \times Q_{\pTwo/}) \cup (Q_{\pTwo/} \times Q_{\pOne/})}
        \exists x^- \exists y^- \bigl( \mathrm{succ}_p(x^-,x) \wedge \mathrm{succ}_q(y^-,y) \wedge \varphi'_{pq}(x^-, y^-) \bigr).
    \]

    \vspace*{-8mm}
\end{proof}

We show the main result, Theorem~\ref{maintheorem}, by an analysis of parity games over the prefix-recognizable game graphs $G(\CA)$ associated with deterministic parity $\NN$-MSO-automata $\CA$.
We use the fact that these graphs have a special format, which allows a simpler analysis:
We say that a prefix-recognizable game graph $(V_1, V_2, E, c)$ with $V_1, V_2 \subseteq \BBB^*$ is \emph{normal} if $\BBB$ can be decomposed as $\AAA \cup Q_\pOne/ \cup Q_\pTwo/$ such that $V_1 = \AAA^* Q_\pOne/$, $V_2 = \AAA^* Q_\pTwo/$, and the color $c(v)$ of a vertex $v = wq$ (with $w \in \AAA^*$ and $q \in Q_\pOne/ \cup Q_\pTwo/$) only depends on $q$.
Clearly, the game graph $G(\CA)$ associated with a given deterministic parity $\NN$-MSO-automata $\CA$ has this property, with $\AAA = \{1\}$.

We shall present a complete solution of parity games over normal prefix-recognizable game graphs in the sense that a certain format of memoryless winning strategy for the respective winner is established, namely as an MSO-definable strategy.
Over a prefix-recognizable game graph defined by MSO-formulas in a tree structure $\CT_{\BBB^*}$, a memoryless strategy $\sigma$ is MSO-definable, say for Player~II, if there is an MSO-formula $\chi_\sigma(x,y)$ over $\CT_{\BBB^*}$ that specifies for each vertex $x$ of the winning region of  Player~II the vertex~$y$ to which Player~II should move according to $\sigma$.
It may be useful to state this result separately from our main theorem~\ref{maintheorem}.

\begin{theorem}\label{mainprefrec}
    Consider a normal prefix-recognizable parity game graph $G$, defined by MSO-formu\-las over a tree $\CT_{\BBB^*}$.
    \begin{enumerate}[label=(\alph*)]
        \item\label{mainprefrec-part-a} The winning regions of the two players are MSO-definable in $\CT_{\BBB^*}$ and hence decidable.
 \eject
        \item\label{mainprefrec-part-b} For a given initial vertex $u_0$ in the winning region of Player~II (or Player~I), there is an MSO-formula $\chi(x,y)$ that defines (again in $\CT_{\BBB^*}$) a memoryless winning strategy from $u_0$ for Player~II (or Player~I, respectively).
    \end{enumerate}
\end{theorem}

Let us first show part~\ref{mainprefrec-part-a}.
We recall the following result of~\cite{Wal02} on definability of winning regions:
\begin{proposition}
    For a parity game over a game graph $G = (V_1, V_2, E, c)$ with coloring function $c\colon V_1 \cup V_2 \rightarrow \{0, \ldots, d\}$,
    the winning regions of the two players are MSO-definable in the structure $G' = (V_1, V_2, E, C_0, \ldots, C_d)$ with $C_i = \{u \in V_1 \cup V_2 \mid c(u)=i\}$ (for $i \in \{0,\dotsc,d\}$) by MSO-formulas $\varphi_\pOne/(x)$ and $\varphi_\pTwo/(x)$.
\end{proposition}
For a parity game graph that can be defined by MSO-formulas over a tree $\CT_{\BBB^*}$, the MSO-formulas $\varphi_\pOne/(x)$ and $\varphi_\pTwo/(x)$ over $G'$ can be transformed into MSO-formulas over~$\CT_{\BBB^*}$.
Combined with the fact that the MSO-theory of $\CT_{\BBB^*}$ is decidable, this yields part~\ref{mainprefrec-part-a} of Theorem~\ref{mainprefrec}.

Part~\ref{mainprefrec-part-a} of Theorem~\ref{mainprefrec} implies the first claim of the main theorem~\ref{maintheorem}:
The winner of the game $\Gamma(L_\CA)$ can be determined effectively.

The subsequent three subsections~\ref{extending} to \ref{transferring} will develop the proof of part~\ref{mainprefrec-part-b} of Theorem~\ref{mainprefrec}.
The last subsection will then transform the MSO-formula $\chi(x,y)$ into a transducer as required in Theorem~\ref{maintheorem}.

Results similar to Theorem~\ref{mainprefrec} were shown by Cachat~\cite{cachat_uniform_2003} and by Kupferman, Piterman, and Vardi~\cite{KPV10}.
Already in~\cite{cachat_uniform_2003}, a reduction of games over prefix-recognizable graphs to games over pushdown graphs appears, similar to the construction given in Section~\ref{extending} below.
We present here a sharpened form in which the corresponding winning strategies are shown to be MSO-definable.
In~\cite{KPV10}, the solution of games over prefix-recognizable graphs relies on the existence of a suitable finite partition of the prefix rewriting rules.
A use of such a partition is avoided in the construction below.

\subsection{Extending a prefix-recognizable to a pushdown game graph}\label{extending}

The aim of this section is to modify a given normal prefix-recognizable game graph $G$ (such as $G(\CA)$) to a form where strategies and winning strategies can be defined essentially by labelling the graph vertices with labels from a finite alphabet.
In a prefix-recognizable game graph, where vertices may have infinite out-degree, a strategy would have to specify a ``next vertex'' among a possibly infinite set of neighbors of a given vertex; so an infinite label alphabet such as $\NN$ seems necessary.
We shall construct from $G$ a ``pushdown game graph'' $H$ where this problem is avoided and in which plays over $G$ can be simulated.

\begin{definition}
    A \emph{pushdown parity game system} over the finite pushdown alphabet $\AAA$ is a tuple $(P_\pOne/, P_\pTwo/, \Delta, c)$
    where $P_\pOne/, P_\pTwo/$ are disjoint finite sets of states, $c\colon (P_\pOne/ \cup P_\pTwo/) \to \{0,\dotsc,d\}$ is a coloring, and $\Delta$ is a finite set of pushdown rules of the form $ap \rightarrow vq$ with $p, q \in P_\pOne/ \cup P_\pTwo/$ and with $a \in \AAA, v \in \AAA^*$ or with $a = \bot, v  \in \bot \AAA^*$ (where $\bot$ is silently added as an extra symbol that is not in~$\AAA$).%
    \footnote{
        In the literature,
        the application of pushdown rules is often done by prefix-rewriting; we use here suffix-rewriting since it fits better to the view that pushdown rules describe steps from one tree node to another.
    }
\end{definition}
Note that these rules define a more general format of pushdown game than sometimes found in the literature:
We do no more require a change from I-states to II-states and conversely; it is possible to  have sequences of moves executed by one player only.
A pushdown parity game system $(P_\pOne/, P_\pTwo/, \Delta, c)$ over $\AAA$ induces a parity game graph
$H = (V_1, V_2, E, c')$ where $V_1 = \AAA^* P_\pOne/$, $V_2 = \AAA^* P_\pTwo/$, $c'(w p) = c(p)$, and the edge relation is
\[
    E = \{(w a p, w v q) \mid (ap \rightarrow vq) \in \Delta \text{ with } a \in \AAA, w \in \AAA^*\} \cup
    \{(p, v q) \mid (\bot p \rightarrow \bot vq) \in \Delta\}.
\]
For a vertex $w p \in V_1 \cup V_2$, we refer to $w \in \AAA^*$ as the stack content and to $p \in P_\pOne/ \cup P_\pTwo/$ as the control state of the vertex.
Note that we use the symbol $\bot$ in pushdown rules to represent the bottom of the stack, but for technical convenience, we do not include it in the stack content of the vertices;
so, for instance, the move from vertex $p$ to vertex $aq$ is induced by the pushdown rule $\bot p \rightarrow \bot aq$.

 \medskip
A game graph derived in this way is called a \emph{pushdown parity game graph} over the alphabet $\CCC \coloneqq \AAA \cup P_\pOne/ \cup P_\pTwo/$.
Its vertices are words over the alphabet $\CCC$ and thus nodes of the tree $\CT_{\CCC^*}$.
It is clear that the relation $E$ is definable in $\CT_{\CCC^*}$ by an MSO-formula (even a first-order formula).
Thus we obtain immediately that every pushdown parity game graph is prefix-recognizable.
However, in contrast to prefix-recognizable graphs in general, the out-degree of the vertices in a pushdown graph is always bounded by the number of pushdown rules and hence finite.

We shall extend a given normal prefix-recognizable parity game graph $G = (U_1, U_2, E, c)$ into a pushdown parity game graph $H = (U'_1, U'_2, E', c')$ in such a way that winning strategies over $H$ can be transferred back to $G$.
As mentioned earlier, the idea is to simulate the edges in $E$ by sequences of ``micro-steps''.
For the realization of these micro-steps auxiliary vertices are introduced; we shall have $U_1 \subseteq U'_1$, $U_2 \subseteq U'_2$.
Vertices of $H$ that do not belong to $G$ are called
\emph{proper $H$-vertices}.
All proper $H$-vertices will have color 0 or 1.
We assume that all $G$-vertices have a color greater than 1, which can always be achieved by shifting their colors by 2.
So in a play that visits $G$-vertices infinitely often, the colors of these $G$-vertices determine the winner.
Among the proper $H$-vertices we shall introduce special ``sink vertices'', in which a play will loop forever; there is one such vertex $u_\text{win}^\pOne/$ in which Player~I will win (it receives color 1) and a vertex $u_\text{win}^\pTwo/$ in which Player~II will win (its color being 0).

 \medskip
To prepare for the construction of $H$, we consider a representation of the edge relation $E$ of $G$ in terms of finite automata.
The vertex sets of $G$ are of the form $U_1 = \AAA^* Q_\pOne/$ and $U_2 = \AAA^* Q_\pTwo/$ for some finite alphabet $\AAA$ and finite sets of control states $Q_\pOne/$, $Q_\pTwo/$,
so the vertices are words over the alphabet $\BBB = \AAA \cup Q_\pOne/ \cup Q_\pTwo/$.
Since $G$ is prefix-recognizable, the edge relation $E$ is defined by an MSO-formula $\psi(x,y)$ over $\CT_{\BBB^*}$.
We apply Proposition \ref{UVWproposition} to $\psi(x,y)$, obtaining regular sets $U_{(i)}, V_{(i)}, W_{(i)} \subseteq \BBB^*$ ($i = 1, \ldots, k$) such that
\[
    \CT_{\BBB^*} \models \psi[w_1  p, w_2  q] \ \ \mbox{iff} \ \bigvee_{i = 1, \ldots, k} \exists u \in U_{(i)} \ \exists v \in V_{(i)} \ \exists w \in W_{(i)}
    (wu = w_1 p \wedge wv = w_2 q).
    \tag{$+$}\label{appliedUVW}
\]
Let $\CU_i, \CV_i, \CW_i$ be deterministic finite automata (with disjoint state sets) recognizing $U_{(i)}, V_{(i)}, W_{(i)}$, respectively ($i = 1, \ldots, k$).

 \medskip
Let $S$ be the union of the state sets of these finite automata and introduce two copies $S_1 \coloneqq \{s_\pOne/ \mid s \in S\}, S_2 \coloneqq \{s_\pTwo/ \mid s \in S\}$ of this set.
Let $S_\pOne/ \coloneqq S_1 \cup \{s_\text{win}^\pTwo/ \}$ and $S_\pTwo/ \coloneqq S_2 \cup \{s_\text{win}^\pOne/\}$ where $s_\text{win}^\pOne/, s_\text{win}^\pTwo/$ are two extra states;
these are used for the representation of the sink vertices $u_\text{win}^\pOne/ \coloneqq s_\text{win}^\pOne/$ and $u_\text{win}^\pTwo/ \coloneqq s_\text{win}^\pTwo/$.
We define the vertex sets of $H$ as $U'_1 \coloneqq \AAA^* (Q_\pOne/ \cup S_\pOne/)$ and $U'_2 \coloneqq \AAA^* (Q_\pTwo/ \cup S_\pTwo/)$.
(This is a liberal definition; not every vertex is actually needed in the construction of $H$ from $G$; each such isolated vertex we equip with a self-loop and make it a sink.)
So the vertices of $H$ are words over the alphabet
\[ \CCC = \AAA \cup Q_\pOne/  \cup Q_\pTwo/ \cup S_\pOne/ \cup S_\pTwo/. \]
The vertices with control state in $S_\pOne/$ have color 0, those with control state in $S_\pTwo/$ have color 1.
(So Player~II, for instance, will lose a play in which she stays in her proper $H$-vertices forever.)

 \medskip
Let us now define the edges of $H$, using the condition (\ref{appliedUVW}) that describes the $G$-edges we want to simulate.
First we give an informal description, explaining how to perform micro-steps that simulate an edge from a vertex $w_1 p$ of $G$, belonging to Player~II,
to another vertex $w_2 q$ of $G$, belonging to Player~I.
(The case that the players are exchanged is handled in complete analogy.)
At the vertex $w_1 p$ in $H$, Player~II first chooses an index $i$ according to (\ref{appliedUVW}); we write $(U,V,W)$ for $(U_{(i)}, V_{(i)}, W_{(i)})$ and use deterministic finite automata $\CU, \CV, \CW$ recognizing $U, V, W$, respectively.
Three phases of micro-steps are executed:
First Player~II removes a suffix $u'$ of the stack content $w_1$ (one letter at a time) while verifying that $u'p$ belongs to $U$, then it is checked whether the remaining stack content $w$ belongs to $W$, and finally Player~II appends a word $v'$ to the stack and switches to a control state $q$ such that $v'q$ belongs to $V$, thus reaching the vertex $w_2 q$ with stack content $w_2 = wv'$.
Since the tests whether $u'p \in U$ and whether $w \in W$ are carried out in reverse order, we change the automata $\CU, \CW$ to $\CU', \CW'$  accordingly.

 \medskip
Let us give more details, describing the pushdown rules that define the edges of $H$, again for the case that we start from a $G$-vertex belonging to Player~II.
At such a vertex, Player~II will have the possibility to pop letters from the stack while running the automaton $\CU'$:
For each $p \in Q_\pTwo/$ a pushdown rule $p \rightarrow s_\pTwo/$ (formally $a p \rightarrow a s_\pTwo/$ for each $a \in \AAA \cup \{\bot\}$) is introduced, where $s$ is the state of $\CU'$ that is reached from the initial state by reading the input $p$.
This rule realizes the change from a $G$-vertex to a proper $H$-vertex.
For each transition of $\CU'$ from a state $s'$ via an input $a \in \AAA$ to a state $s''$,
the pushdown rule $a s'_\pTwo/  \rightarrow s''_\pTwo/$ is added.
To avoid deadlocks, we also add the rule $\bot s'_\pTwo/  \rightarrow \bot s'_\pTwo/$ for each state $s'$ of $\CU'$.
Whenever an accepting state $s^\text{acc}$ of $\CU'$ is reached, say in a vertex $w s^\text{acc}_\pTwo/$,
Player~II has the choice to let Player~I take her turn, using a pushdown rule $s^\text{acc}_\pTwo/ \rightarrow t_\pOne/$ where $t$ is the initial state of~$\CW'$.
Player~I now has two choices:
\begin{enumerate}
    \item either to check whether the remaining stack content $w$ belongs to $W$;
        this is done by simulating $\CW'$ while emptying the stack, and proceeding to the sink vertex $u_\text{win}^\pTwo/$ of color 0 if $\CW'$ accepts or to the sink vertex $u_\text{win}^\pOne/$ of color 1 if $\CW'$ does not accept,
    \item or to let Player~II proceed again, in which case Player~II will push letters from $\AAA$ onto the stack (one at a time), simulating $\CV$.
        Suppose that at some point, Player~II has pushed a word $v'$ and $\CV$ can reach an accepting state from its current state $r$ by reading a letter $q \in Q_\pOne/ \cup Q_\pTwo/$ (which means that the word $v'q$ belongs to $V$).
        Then Player~II can choose to leave the domain of proper $H$-vertices and enter the $G$-vertex $wv'q$ by applying the pushdown rule $r_\pTwo/ \rightarrow q$.
        Note that Player~II will lose the play (by staying with vertices of color 1) if she continues to increase the stack indefinitely.
\end{enumerate}
It is easy but a little tedious to supply the corresponding pushdown rules in detail;  we do not present this list here.
The explanation of the micro-steps induced by the automaton $\CU'$ shows how to construct these pushdown rules.
(The careful reader might note that the case where the languages $U$ and $V$ contain the empty word requires extra attention.
A simple preprocessing of $U, V, W$, not specified here, will take care of this case.)
It should be clear how to proceed in a dual way to simulate edges of $G$ from vertices that belong to Player~I instead of Player~II.

Having completed the construction of the game graph $H$ from the game graph $G$,
we show in which sense the parity game over $G$ is simulated by the parity game over $H$.
\begin{proposition}\label{gameequivalence}
    Let $G$ be a normal prefix-recognizable parity game graph and let $H$ be the corresponding pushdown parity game graph as constructed above.
    A memoryless winning strategy for Player~II (or Player~I) in the parity game over $H$ from a $G$-vertex $u_0$ can be transformed into a memoryless winning strategy for Player~II (or Player~I, respectively) in the parity game over $G$ from $u_0$.
\end{proposition}
\begin{proof}
    We treat the case of Player~II, the case of Player~I is analogous.

 \medskip
    Assume that Player~II has a memoryless winning strategy $\tau$ over $H$ from $u_0$.
    Whenever a $G$-vertex $w_1 p$ belonging to Player~II is reached, the strategy $\tau$ determines a path of micro-steps that leads to another $G$-vertex, based on a rule $(U,V,W)$ appearing in the equivalence (\ref{appliedUVW}):
    First, the strategy selects a sequence of micro-steps that remove letters from the stack.
    Since $\tau$ is a winning strategy, it will not stay forever within the domain of proper $H$-vertices belonging to Player~II (which have color~1) as this would result in a loss.
    Thus the micro-steps selected by $\tau$ will remove a suffix $u'$ of the stack content such that $u'p \in U$ and then enter a vertex where Player~I can choose to inspect the remaining stack content $w$ or to let Player~II proceed.
    Note that $w \in W$; otherwise, Player~I could win by inspecting the stack, contradicting the assumption that $\tau$ is a winning strategy.
    Now consider the vertex that is reached if Player~I lets Player~II proceed.
    From this vertex, the strategy $\tau$ selects a sequence of micro-steps that push a new suffix onto the stack.
    This sequence will not be infinite, as this would result in a loss for Player~II due to the colors of the proper $H$-vertices.
    Thus the strategy will push some finite suffix $v'$ and switch to a control state $q \in Q_\pOne/ \cup Q_\pTwo/$ such that $v'q \in V$, leading to the vertex $w_2 q$ with stack content $w_2 = wv'$.
    So the strategy $\tau$ determines a path from $w_1 p$ to $w_2 q$ that corresponds to an edge $(w_1 p, w_2 q)$ of $G$.

    Now let $\sigma$ be the memoryless strategy for Player~II over $G$ that chooses the edges $(w_1 p, w_2 q)$ as just described.
    For each play over $G$ from $u_0$ that results from the strategy $\sigma$ there is a corresponding play over $H$ that contains the same sequence of $G$-vertices (visiting proper $H$-vertices in between) and that results from the winning strategy $\tau$ over $H$.
    The maximal color occurring infinitely often is the same in both plays because the colors of the $G$-vertices exceed those of the proper $H$-vertices.
    Since the play over $H$ is won by Player~II, so is the play over $G$.
\end{proof}

\subsection{Constructing a regular winning strategy in a pushdown game}\label{regularstrategy}

After the transformation of a normal prefix-recognizable game graph $G$ into a pushdown game graph $H$, the next step on the way to an MSO-definable winning strategy over $G$ is the construction of a ``definable'' winning strategy over $H$.
For this, we define a simple representation of memoryless strategies on pushdown game graphs and show that the winning strategies for a player in a pushdown parity game are recognizable by appropriate finite-state automata.
If the set of winning strategies is nonempty,
we can extract a ``regular'' winning strategy.

Let $(P_\pOne/, P_\pTwo/, \Delta, c)$ be a pushdown parity game system over $\AAA$ and let $H = (V_1, V_2, E, c)$ be the associated pushdown parity game graph; its vertices are words in $\AAA^* (P_\pOne/ \cup P_\pTwo/)$.
A memoryless strategy (for Player~II, say) has to specify which rule of $\Delta$ to apply in a vertex $w p$ with $p \in P_\pTwo/$.
We do this by associating with the word $w \in \AAA^*$ an  \emph{instruction}, telling Player~II for each state $p \in P_\pTwo/$ which rule to use when visiting vertex $w p$.
So an instruction for Player~II is a function   $\iota\colon P_\pTwo/  \rightarrow  \Delta$.
Let $\CI_2$ be the (finite!) set of possible instructions for Player~II.
Thus a (memoryless) strategy $\tau$ for Player~II is fixed by a function  $f_\tau \colon \AAA^* \rightarrow \CI_2$, in other words by labelling the nodes of the tree $t_{\AAA^*}$ with values in $\CI_2$; denote this labelled tree by $t_\tau$.
The label $t_\tau(w)$ indicates that Player~II should apply the rule $[t_\tau(w)](p)$ in vertex $w p$.
We call a strategy \emph{regular} if it is defined by a regular $\CI_2$-labelled tree.
Dually, we proceed for Player~I with the obvious changes.

In the following  proposition, which is a variant of a result of~\cite{KPV10}, we use concepts and facts on tree automata as mentioned in Subsection \ref{preliminaries} above.

\begin{proposition}
    Let $(P_\pOne/, P_\pTwo/, \Delta, c)$ be a pushdown parity game system over $\AAA$, inducing a parity game graph $H$.
    Let $\CI_2$ be the set of instructions
    $\iota\colon P_\pTwo/  \rightarrow  \Delta$ of Player~II.
    \begin{enumerate}[label=(\alph*)]
    \itemsep=0.9pt
        \item The memoryless winning strategies for Player~II in the parity game over $H$ from a given initial vertex $v_0$ form a set of $\CI_2$-labelled trees, labelling $t_{\AAA^*}$, that is recognizable by a parity tree automaton.
        \item If there is a winning strategy for Player~II in the parity game over $H$ from a given vertex $v_0$, then there is an (effectively constructible) regular winning strategy.
    \end{enumerate}
    Analogously, these statements also hold for Player~I, with $\CI_2$ being replaced by the set $\CI_1$ of instructions $\iota\colon P_\pOne/  \rightarrow  \Delta$ of Player~I.
\end{proposition}

\begin{proof}
    It suffices to show part (a); part (b) then follows directly from the solution of the non-emptiness problem for parity tree automata (see Theorem \ref{rabinbasis}).

 \medskip
    We treat the case of Player~II; for Player~I one proceeds analogously.
    We construct a universal two-way parity tree automaton $\CB$ that works on a $\CI_2$-labelled tree $t$ whose nodes are the words in~$\AAA^*$.
    We assume that the automaton can recognize the root of the tree,
    and for any other vertex $wa$ (with $w \in \AAA^*$ and $a \in \AAA$), in the move upwards to vertex $w$ the letter $a$ is recognized;
    so, in particular, the automaton can inspect the last letter of its current vertex (by moving up one node and then moving down again).

    The states of the automaton will be the states in $P_\pOne/ \cup P_\pTwo/$ and some auxiliary states (e.g., for initialization).
    While traversing  $t$ up and down,  $\CB$ simulates
    \emph{all} plays with start vertex $v_0$ that are compatible with the labels of $t$.
    The automaton works in two-way mode and assumes state $p$ on position $w \in \AAA^*$ of the tree if in a play starting in $v_0$ the vertex $w p$ of $H$ is reached.
    If the current control state $p$ (and hence the current vertex $w p$) belongs to Player~II, the automaton $\CB$ executes the change of control state and of position  as given by the pushdown rule $[t(w)](p)$.
    If the control state $p$ belongs to Player~I, the automaton at position $w$ uses universal branching and executes in parallel
    \emph{all} changes of control state and of position as allowed by the pushdown rules in $\Delta$.
    (When a rule $ap \rightarrow vq$ involves a word $v$ of length greater than $0$, auxiliary states are used to implement the application of the rule.)
    Before starting the simulation of a play, the automaton performs a phase of initialization:
    If the given start vertex is $v_0 = w_0 p_0$, the automaton walks to node $w_0$ of the tree and assumes state $p_0$.
    The states of the tree automaton are equipped with colors according to the coloring $c$ provided by the pushdown parity game system, as far as they belong to $P_\pOne/ \cup P_\pTwo/$; all other states  receive color $0$.

    It is clear that the tree automaton $\CB$, using the parity condition for acceptance as defined,
    accepts precisely the memoryless winning strategies for Player~II from the given start vertex.
\end{proof}

\subsection{Transferring the regular winning strategy}\label{transferring}

Given a normal prefix-recognizable parity game graph $G$ and an initial vertex $u_0$, we can construct a corresponding pushdown parity game graph $H$ (see Section~\ref{extending}) and construct a regular winning strategy over $H$ from $u_0$ (see Section~\ref{regularstrategy}).
We shall now transform this regular winning strategy over $H$ into an MSO-definable winning strategy over $G$, thereby completing the proof of Theorem~\ref{mainprefrec}.
First, we represent the regular strategy over $H$ in terms of MSO-logic.

\begin{proposition}
    Given a regular strategy $\tau$ for one of the players on a pushdown game graph $H$ over the alphabet $\CCC$, there is an MSO-formula $\varphi_\tau(x,y)$ defining $\tau$ in $\CT_{\CCC^*}$.
\end{proposition}
\begin{proof}
    Let $(P_\pOne/, P_\pTwo/, \Delta, c)$ be the pushdown parity game system over $\BBB$ defining $H$.
    (Note that $\CCC = \BBB \cup P_\pOne/ \cup P_\pTwo/$.)
    Suppose that $\CA_\tau$ is a deterministic finite automaton defining $\tau$, which means that after reading a stack content (i.e., a word) $u \in \BBB^*$, the automaton outputs the instruction $f_\tau(u)$ according to the strategy~$\tau$.
    For each rule $\varrho \in \Delta$ let $\varphi_\varrho(x,y)$ be an MSO-formula (we can take even a first-order formula) expressing in $\CT_{\CCC^*}$ that from vertex $x \in t_{\CCC^*}$ the application of $\varrho$ yields vertex $y$.
    Next we define the change of a vertex $x$ to a vertex $y$ that results from the application of the rule $\varrho$ as fixed by the strategy $\tau$:
    From $\CA_\tau$ we can find MSO-formulas $\varphi_{\iota, p}(x)$ stating that $x$ is of the form $u p$ such that $\CA_\tau$ produces $\iota$ after reading $u$.
    Now the change of a vertex $x$ of $H$ to $y$ according to $\tau$ is expressed by
    \[
        \varphi_\tau(x,y) \coloneqq  \bigvee_{\iota \in \CI_2, \iota(p) = \varrho} (\varphi_{\iota, p}(x) \wedge \varphi_\varrho(x, y)).
    \]

    \vspace*{-8mm}
\end{proof}

The resulting MSO-definable winning strategy over the pushdown game graph $H$ can now be transferred to the normal prefix-recognizable game graph $G$:

\begin{proposition}\label{winningstrat}
    Let $G$ be a normal prefix-recognizable game graph over the alphabet $\BBB$ and let $H$ be the corresponding pushdown parity game graph $H$ over the alphabet $\CCC \supseteq \BBB$, as described in Section~\ref{extending}.
    Given a winning strategy $\tau$ from a $G$-vertex $u_0$ in $H$ for one of the players that is defined by an MSO-formula $\varphi_\tau(x,y)$ over $\CT_{\CCC^*}$, there is a winning strategy $\sigma$ from the vertex $u_0$ in $G$ for the same player that is defined by an MSO-formula $\chi_{\sigma}(x,y)$ over $\CT_{\BBB^*}$.
\end{proposition}
\begin{proof}
    We treat the case of Player~II (and proceed analogously for Player~I).
    It suffices to define the desired winning strategy $\sigma$ in $G$ by an MSO-formula $\chi^{\CCC}_{\sigma}(x,y)$ over $\CT_{\CCC^*}$;
    since the strategy $\sigma$ is in fact a function over $\BBB^*$, this formula can then be transformed into a formula $\chi_{\sigma}(x,y)$ over $\CT_{\BBB^*}$ by Proposition~\ref{alphabetrestriction}.
    We now present the construction of $\chi^{\CCC}_{\sigma}(x,y)$, using the fact that MSO-logic allows to express the transitive closure of MSO-definable binary relations.

    Recall that the vertex sets $V_1$ and $V_2$ of $H$ belonging to Player~I and Player~II, respectively, are definable in $\CT_{\CCC^*}$, say by the formulas $\varphi_1(x)$ and $\varphi_2(x)$.
    Similarly, the set of $G$-vertices is definable in $\CT_{\CCC^*}$ by a formula $\varphi_G(x)$.
    We express that from the $G$-vertex $x$ (in the game graph $H$) belonging to Player~II we reach the $G$-vertex $y$ by steps according to $\varphi_\tau$, disallowing intermediate vertices of $G$;
    the resulting strategy is a winning strategy, as described in the proof of Proposition~\ref{gameequivalence}.

 \medskip
    For this we restrict the formula $\varphi_\tau(x,y)$ to proper $H$-vertices by defining
    $\psi_\tau(y_1,y_2) \coloneqq \neg \varphi_G(y_1) \wedge \neg \varphi_G(y_2) \wedge \varphi_\tau(y_1, y_2)$.
    The desired formula $\chi^{\CCC}_{\sigma}(x,y)$ expresses that $\tau$ establishes a path that leads from $x$ to a proper $H$-vertex $z_1$ (by a single micro-step), from there by a sequence of micro-steps to a last proper $H$-vertex $z_2$ and from there (again by a single micro-step) to $y$.
    The path from $z_1$ to $z_2$ is divided into three parts: (1) first a sequence of micro-steps according to the strategy $\tau$ leading to a vertex $y_1$ of Player~I, (2) then the micro-step by which Player~I lets Player~II continue immediately in vertex $y_2$, and (3) from $y_2$ another sequence of micro-steps according to $\tau$ leading to $z_2$.
    (The option of Player~I to inspect the stack content of vertex $y_1$ is not relevant for the path from $x$ to $y$.)
    So the formula $\chi^{\CCC}_{\sigma}(x,y)$ has the form
    \[
        \chi^{\CCC}_{\sigma}(x,y) \coloneqq \varphi_2(x) \wedge \varphi_1(y) \wedge \varphi_G(x) \wedge \varphi_G(y) \wedge \chi_0(x,y),
    \]
    where $\chi_0(x,y)$ is the formula
    \[
        \exists z_1, z_2, y_1, y_2
        ( \varphi_{\tau}(x, z_1) \wedge \theta_{(1)}(z_1, y_1) \wedge \theta_{(2)}(y_1, y_2) \wedge \theta_{(3)}(y_2, z_2) \wedge \varphi_{\tau}(z_2, y) ).
    \]
    For the formula $\theta_{(1)}(z_1,y_1)$ we can take $\psi_{\tau}^*(z_1, y_1)$, similarly for $\theta_{(3)}(y_2,z_2)$.
    Finally, $\theta_{(2)}(y_1,y_2)$ just has to express the change of control state when Player~I chooses to let Player~II continue again.
\end{proof}

\subsection{Constructing a transducer}

In this subsection we return to the original game $\Gamma(L_\CA)$ defined by a given deterministic parity $\NN$-MSO-automaton $\CA$.
So far, we have shown how a memoryless winning strategy $\sigma$ for the winner of the parity game over the normal prefix-recognizable game graph $G(\CA)$ from the initial vertex $(q_0)_\pOne/$ can be constructed, represented by an MSO-formula $\chi_\sigma(x,y)$.
We assume that Player~II is the winner, the case of Player~I can he handled analogously.
To finish the proof of the main theorem,  it remains to show that the formula $\chi_\sigma(x,y)$ can be transformed into a transducer realizing a winning strategy for the original game $\Gamma(L_\CA)$.

This transducer must generate, number by number, an $\omega$-word  $n_0 n_1 n_2 \ldots \in \NN^\omega$ for an $\omega$-word $m_0 m_1 m_2 \ldots \in \NN^\omega$ such that the deterministic $\NN$-MSO parity-automaton $\CA$  accepts the combined $\omega$-word $m_0 n_0 m_1 n_1 m_2 n_2 \ldots$.
A step of the transducer involves reading an input $m_i$ and producing an output $n_i$, corresponding to two steps of $\CA$.
In order to simulate the run of $\CA$ on the combined sequence, the transducer will have transition formulas that describe the change of configuration by the previous output together with the current input; for the first input (where a previous output is lacking) only the change of the initial configuration by the first input is described.

The output that is produced in a given configuration will be defined to be compatible with the strategy $\sigma$:
For a configuration $(p,i)$ there is a corresponding vertex of Player~II in the game graph $G(\CA)$, namely $1^i p_\pTwo/$, and the strategy $\sigma$ prescribes a move to an adjacent vertex $1^j q_\pOne/$ representing a configuration $(q,j)$.
The output $n$ of the transducer in the configuration $(p,i)$ will be chosen in such a way that the transition of $\CA$ from $(p,i)$ via $n$ leads to the configuration $(q,j)$.

For the formal construction of the transducer, we need a technical preparation:
We have to translate the formula $\chi_\sigma(x,y)$ interpreted in $\CT_{\BBB^*}$ into formulas interpreted in $\CN_\mathrm{succ}$, as they appear in $\NN$-MSO-automata and $\NN$-MSO-transducers.

\begin{lemma}\label{translation}
    Let $\BBB = \{1\} \cup \{q_\pOne/ \mid q \in Q\} \cup \{q_\pTwo/ \mid q \in Q\}$, where $Q$ is a finite set.
    For a formula $\chi(x,y)$ interpreted in $\CT_{\BBB^*}$,
    one can construct formulas $\delta_{pq}(x,y)$, for $p,q \in Q$, interpreted in $\CN_\mathrm{succ}$, such that
    \[
        \CT_{\BBB^*} \models \chi[1^m p_\pTwo/, 1^n q_\pOne/] \ \
        \mbox{iff} \  \  \CN_\mathrm{succ} \models \delta_{pq}[m,n].
    \]
\end{lemma}

\begin{proof}
    For $p,q \in Q$, we define $\chi_{pq}(x,y) \coloneqq \exists x',y' \bigl(\mathrm{succ}_{p_\pTwo/}(x,x') \wedge \mathrm{succ}_{q_\pOne/}(y,y') \wedge \chi(x',y') \bigr)$.
    Then we have
    \[
        \CT_{\BBB^*} \models \chi[1^m p_\pTwo/, 1^n q_\pOne/] \ \ \mbox {iff} \ \ \CT_{\BBB^*} \models \chi_{pq}[1^m, 1^n].
    \]
    The MSO-formula $\chi_{pq}$ over $\CT_{\BBB^*}$ can be converted into the desired MSO-formula $\delta_{pq}$ over $\CN_\mathrm{succ}$ in the following way:
    The formula $\chi_{pq}$ defines a relation of words over the alphabet $\{1\} \subseteq \BBB$, so Proposition~\ref{alphabetrestriction} can be applied to obtain an equivalent formula over $\CT_{\{1\}^*}$.
    Since $\CT_{\{1\}^*}$ is isomorphic to $\CN_\mathrm{succ}$, renaming $\mathrm{succ}_1$ to $\mathrm{succ}$ yields the desired formula $\delta_{pq}$.
\end{proof}

Let $\CA = (Q, q_0, (\varphi_{pq}(x, z, y))_{p,q \in Q}, c)$ be the given deterministic parity $\NN$-MSO-automaton.
We now define the desired transducer
\[
    \CC = (Q \cup \{q'_0\}, q'_0, (\nu_{pq}(x,z,y))_{p,  q \in Q}, (\psi_{q}(x,z))_{q \in Q})
\]
for Player~II.
Recall that the entry $z$ in $\nu_{pq}$ will stand for a number contributed by Player~I, while the $z$ in the output formula $\psi_q$ represents the choice of Player~II.
In the following we shall use suggestive variable names for control states and numbers to support readability.

 \medskip
We isolate the initial step (by Player~I) with the formulas
\[
    \nu_{q'_0 p}(x,z,y) \coloneqq  (x = 0) \wedge \varphi_{q_0 p}(x,z,y).
\]
The subsequent output of Player~II is specified (as are the later outputs) using the formulas $\delta^\sigma_{pq}$ that are obtained from $\chi_\sigma$ according to Lemma~\ref{translation}.
More specifically, the output in a given configuration is defined to be the smallest number that is compatible with the change of configuration of the automaton $\CA$ that is determined by the formulas $\delta^\sigma_{pq}$:
\[
    \psi_p(y,z') \coloneqq \bigvee_{q \in Q}   [\exists y' (\delta^\sigma_{pq}(y,y') \wedge \varphi_{pq}(y,z',y') \wedge \neg \exists z^-(z^- < z' \wedge \varphi_{pq}(y, z^-, y')) ].
\]

As explained at the beginning of this subsection, the transition formulas describe the change of configuration induced by the previous output (executed by Player~II) and the current input (contributed by Player~I):
\[
    \nu_{pr}(y,z,y'') \coloneqq  \bigvee_{q \in Q} \exists y' \ [\delta^\sigma_{pq}(y, y') \wedge \varphi_{qr}(y', z,y'')].
\]
With these formulas $\nu_{pr}$ and $\psi_p$,
the transducer $\CC$ realizes a winning strategy for Player~II in $\Gamma(L_\CA)$, assuming that Player~II wins $\Gamma(L_\CA)$.
This concludes the proof of Theorem \ref{maintheorem}.

\begin{remark}\label{numberchoice}
    The formula $\psi_p(y,z')$ in the proof above serves to select the number provided by Player~II to be the smallest one that is compatible with the move to the new configuration as determined by the winning strategy.
    In more general types of games, a corresponding MSO-definable selection of the ``letters'' to be provided by Player~II may be lacking.
    In Section~\ref{variants} (items \ref{word-memory} and \ref{tree-letters}) we mentioned such cases where the present set-up of proof is not applicable.
\end{remark}

\section{Conclusion}

We have proved an analogue of the B\"uchi-Landweber Theorem~\cite{BL69} for games over the alphabet~$\NN$, where winning conditions are subsets of the Baire space.
For the specification of a game we used a natural model of automaton (``$\NN$-MSO-automaton'') in which transitions are defined by MSO-formulas, and for the realization of winning strategies we defined a corresponding model of transducer.
Our result is completely qualitative; we did not enter considerations on complexity.

 \medskip
We see many open problems that are motivated by the present paper.
We mention just four of them.
\begin{itemize}
\item
In our proof we switched from MSO-logic to automata and back at several places, which causes complexity bounds beyond any practical relevance.
Thus a natural question is how to avoid the use of MSO-logic (even in the specification of games and in transducers) to obtain reasonable complexity bounds.
In this context, we would like to mention work of Carayol and Hague~\cite{CarayolH14} which offers a saturation-based approach for the construction strategies in pushdown games with reachability objectives.
It seems interesting to develop this approach and to find a more direct proof of our main result.

\item
The weak closure properties of $\NN$-MSO-automata prohibit an elegant connection to logic,
in contrast to the situation for finite-state $\omega$-automata where we have expressive equivalence with MSO-logic.
One may ask whether there are logical systems that can serve as a framework for the definition of games in the Baire space and still allow an effective solution of Church's synthesis problem.

\item
The deterministic parity $\NN$-MSO-automata define $\omega$-languages in the Borel class BC($\Pi_2$), the Boolean closure of $\Pi_2$.
It would be interesting to present effective criteria for these automata by which one can decide whether the recognized $\omega$-language is located lower in the Borel hierarchy, e.g., whether it is an open set or a $\Pi_2$-set.
For finite-state $\omega$-automata such criteria are known~(\cite{landweber_decision_1969}).

\item
In their nondeterministic version, $\NN$-MSO-automata can recognize $\omega$-languages of higher Borel level than BC($\Pi_2$), the Boolean closure of $\Pi_2$, for example the $\omega$-language that contains the sequences in which some number occurs infinitely often (which gives a proper $\Sigma_3$-set; cf.~\cite{cachat_solving_2002})).
It would be interesting to extend deterministic $\NN$-MSO-automata by acceptance conditions that allow to recognize $\omega$-languages beyond BC$(\Pi_2$) (as in~\cite{Se06} over finite alphabets)
such that for the corresponding games an analogue of the main result of this paper still holds.
\end{itemize}

Finally, we mention that in our main result (Theorem~\ref{maintheorem}) it is essential that only deterministic parity $\NN$-MSO-automata are admitted for the specification of infinite games.
From work of Finkel it follows that nondeterministic parity $\NN$-MSO-automata can define games that are not Borel~\cite{Fi03}, and that for such games already the question of determinacy depends on the assumption of set-theoretic hypotheses~\cite{Fi13}.
Thus severe restrictions on nondeterministic $\NN$-MSO-automata would have to be imposed when an analogue of Theorem~\ref{maintheorem} is envisaged.

\bibliographystyle{fundam}
\bibliography{bibliography}

\end{document}